\documentclass{article}

\PassOptionsToPackage{numbers, compress}{natbib}
\usepackage[preprint]{neurips_2026}

\usepackage[utf8]{inputenc} 
\usepackage[T1]{fontenc}    
\usepackage{hyperref}      
\usepackage{url}            
\usepackage{booktabs}       
\usepackage{amsfonts}       
\usepackage{nicefrac}       
\usepackage{microtype}      
\usepackage{xcolor}         

\definecolor{darkblue}{rgb}{0, 0, 0.5}
\hypersetup{colorlinks=true, citecolor=darkblue,linkcolor=darkblue,urlcolor=darkblue} 
\usepackage{algorithm}
\usepackage{algorithmic}
\usepackage{amsmath}
\usepackage{amssymb}
\usepackage{mathtools}
\usepackage{amsthm}

\theoremstyle{definition} 
\newtheorem{theorem}{Theorem}[section]
\newtheorem{proposition}[theorem]{Proposition}

\theoremstyle{definition}

\theoremstyle{remark}

\usepackage{xspace}
\usepackage{siunitx}
\usepackage{graphicx}
\newcommand{\cgfsoftmax}{CGF-softmax\xspace}

\usepackage{multirow}

\title{CGF-Softmax: A Cumulant-Based Softmax Reformulation for Efficient Inference under Homomorphic Encryption}

\author{%
Hanjun Park$^{1}$\thanks{Equal contribution} \quad
Byeongseo Min$^{1}$\footnotemark[1] \quad
Jiheon Woo$^{1}$ \quad
Min-Wook Jeong$^{2}$ \quad \\
\textbf{Jongho Shin}$^{2}$ \quad
\textbf{Yongwoo Lee}$^{3}$ \quad
\textbf{Young-Sik Kim}$^{4}$ \quad
\textbf{Yongjune Kim}$^{1}$\thanks{Corresponding author} \\
$^{1}$Pohang University of Science and Technology (POSTECH) \quad \\
$^{2}$LG Electronics R\&D Center \quad
$^{3}$Inha University \quad \\
$^{4}$Daegu Gyeongbuk Institute of Science and Technology (DGIST) \quad \\
\texttt{\{parkterry, minbyeongseo, jhwoo1997, yongjune\}@postech.ac.kr} \\
\texttt{\{minwook.jeong, jongho0.shin\}@lge.com} \\
\texttt{yongwoo@inha.ac.kr, ysk@dgist.ac.kr}
}

\begin{document}

\maketitle

\begin{abstract}
Homomorphic encryption (HE) is a prominent framework for privacy-preserving machine learning, enabling inference directly on encrypted data.
However, evaluating softmax, a core component of transformer architectures, remains particularly challenging in HE due to its multivariate structure, the large dynamic range induced by exponential functions, and the costly division operation.
In this paper, we propose CGF-softmax, which reformulates the softmax denominator through the cumulant generating function (CGF).
By eliminating both homomorphic division and explicit maximum subtraction, this reformulation substantially reduces multiplicative depth while preserving key properties of softmax. 
Extensive experiments on Vision Transformers and large language models show that CGF-softmax provides an efficient and accurate approximation of softmax in encrypted inference.
In particular, it achieves inference accuracy close to that of high-depth exact methods, while requiring substantially lower computational cost through reduced multiplicative depth.
\end{abstract}

\section{Introduction}
\label{sec:Introduction}
The rapid advancement of artificial intelligence (AI) has accelerated the adoption of machine learning as a service (MLaaS), where users increasingly rely on remotely hosted models for inference across a wide range of applications.
However, this paradigm inherently raises privacy concerns, as it typically requires exposing sensitive user data to remote infrastructure. 
To address these risks, privacy-preserving machine learning (PPML) has emerged as a critical research area aimed at enabling secure inference without compromising the confidentiality of user data~\cite{Al-Rubaie2019Privacy,Riazi2019deep}. 

One promising direction within PPML is to employ homomorphic encryption (HE)~\cite{Rivest1978On,Gentry2009fully}, which enables computation directly over encrypted data.
A broad line of prior work has leveraged HE to realize HE-based PPML inference pipelines that operate entirely in the encrypted domain.
In practice, HE-based PPML protocols can be realized in either interactive or strictly non-interactive settings.
Interactive protocols--often instantiated via secure multi-party computation (MPC) \cite{Evans2018pragmatic}--require client participation during inference and therefore incur substantial communication overhead.
In contrast, strictly non-interactive HE allows the server to execute inference autonomously once ciphertexts are received, eliminating communication and client-side burden at the cost of increased server-side computation.
In this work, we focus on the strictly non-interactive setting, where avoiding interaction during inference is essential for fully outsourced inference.

A primary challenge of HE-based PPML lies in the implementation of non-polynomial operations.
Most HE schemes natively support only addition and multiplication, rendering the direct evaluation of non-polynomial functions infeasible.
As a result, non-polynomial activation functions are typically approximated using polynomials, which can be evaluated homomorphically.
However, achieving high approximation accuracy generally requires high-degree polynomials, leading to excessive computational latency and noise growth.
This leads to increased multiplicative depth, which consumes the limited level budget and necessitates costly bootstrapping operations.
To reduce this overhead, prior work has often resorted to simplified surrogate functions or low-degree polynomial approximations, at the cost of degraded model accuracy~\cite{Lou2020safenet,Lou2021hemet,Chen2022the,Rho2025encryption}.

This limitation is particularly pronounced in transformer-based architectures~\cite{Vaswani2017attention,Brown2020language,Dosovitskiy2021image,Dubey2024llama}, where the \emph{softmax} function plays a central role in the self-attention mechanism.
Unlike univariate activation functions, softmax has a multivariate structure where each output depends on the entire input vector through a shared normalization term.
Evaluating softmax in a strictly non-interactive HE setting therefore constitutes a major computational bottleneck, as its standard formulation involves the exponential function and division, which are inherently expensive to approximate homomorphically~\cite{Hong2022Secure,Cho2024fast}.
Moreover, numerical stabilization techniques commonly employed in plaintext inference, such as subtracting the maximum input value, require comparison operations that are fundamentally incompatible with HE.

In response to this challenge, existing approaches to homomorphic softmax evaluation largely fall into two categories.
Methods commonly referred to as \emph{softmax approximation} aim to evaluate the softmax function with high precision under HE, but are typically computationally expensive due to high-degree polynomial approximations for the exponential function and division~\cite{Hong2022Secure,Cho2024fast,DeCastro2025encryptedllm,Lim2025tricycle}. 
To reduce this computational cost, an alternative strategy denoted as \emph{softmax replacement} replaces the original softmax function with simpler surrogate functions~\cite{Zimerman2024converting,Zimerman2024power,Park2025Powerformer}.
While such replacements can be effective for relatively simple natural language processing (NLP) tasks, our experiments show that they suffer from substantial accuracy degradation in more demanding settings, including large-scale image classification and complicated NLP tasks. 
These limitations highlight the need for a softmax evaluation method that significantly reduces computational cost while preserving high accuracy on complex tasks. 

In this paper, we propose \textbf{CGF-softmax}, a novel method that \emph{reformulates} the softmax function by leveraging the cumulant generating function (CGF).
CGF-softmax addresses the aforementioned challenges by eliminating the need for homomorphic division as well as explicit maximum subtraction commonly used for numerical stabilization. 
By removing these computationally expensive operations, particularly the comparison-heavy max operation, our method substantially reduces the multiplicative depth compared to existing softmax approximation approaches.
This depth reduction is critical as it minimizes the need for bootstrapping, directly translating to significant savings in computational cost.
We further derive key theoretical properties of CGF-softmax and provide an asymptotic bound on its approximation error relative to the exact softmax function.
Experimental results demonstrate that CGF-softmax achieves a significant reduction in the level consumption, outperforming the softmax approximation baseline~\cite{Cho2024fast}.
Moreover, CGF-softmax consistently maintains high accuracy across diverse tasks, achieving less than a \SI{1}{\%} accuracy drop for Vision Transformers (ViT/DeiT)~\cite{Dosovitskiy2021image,Touvron2021training} on ImageNet-1k~\cite{Deng2009imagenet} as well as large language model (LLaMA-3.2-1B~\cite{Meta2024meta}) on Clinc150~\cite{Larson2019evaluation}, Banking77~\cite{Casanueva2020efficient} and SST-2~\cite{Wang2018glue}.
These benchmarks are particularly challenging in homomorphic inference, where existing softmax replacement methods exhibit significant degradation.

Our contributions are summarized as follows:
\begin{itemize}
\vspace{-1mm}
    \item \textbf{CGF-softmax Formulation:} We propose CGF-softmax, a novel reformulation of the softmax function that eliminates the need for homomorphic division and maximum subtraction, thereby substantially reducing computational cost.
    \item \textbf{Theoretical Analysis:} We provide a theoretical analysis establishing key properties of CGF-softmax and an asymptotic bound on its approximation error.
    \item \textbf{Experimental Validation:} We demonstrate through extensive experiments on Vision Transformers and a large language model that CGF-softmax achieves inference accuracy close to exact softmax while significantly reducing inference cost under homomorphic encryption.
\end{itemize}
\newpage
\section{Background and Related Works} \label{sec:background}
\subsection{Fully Homomorphic Encryption}
Fully homomorphic encryption (FHE) is a cryptographic primitive that enables computations to be performed directly on encrypted data without decryption.
Among existing schemes, we utilize the CKKS scheme~\cite{Cheon2017homomorphic}, which supports approximate arithmetic operations over encrypted real numbers.

Specifically, CKKS employs a packing technique that encrypts a vector of values into a single ciphertext. 
We denote the number of these values, i.e., the number of slots in a ciphertext, as $s$.
Supported homomorphic operations include element-wise addition ($\mathsf{Add}$), cyclic rotation ($\mathsf{Rot}$), and multiplication.
Regarding multiplication, we distinguish between \emph{Ciphertext--Plaintext Multiplication} ($\mathsf{PMult}$), which involves a plaintext constant, and the computationally expensive \emph{Ciphertext--Ciphertext Multiplication} ($\mathsf{CMult}$).

As a leveled FHE scheme, CKKS assigns a finite budget of multiplicative levels $L$ to each ciphertext.
Each $\mathsf{CMult}$ (and non-integer $\mathsf{PMult}$) consumes one level, and the cumulative consumption determines the circuit \emph{depth}.
Once the level budget is exhausted, \emph{bootstrapping} ($\mathsf{Boot}$) is required to refresh the ciphertext, enabling further computation at the cost of significant overhead.
Consequently, we evaluate HE algorithmic efficiency based on the multiplicative depth and the counts of dominant operations, specifically $\mathsf{CMult}$ and $\mathsf{Rot}$.

\subsection{Homomorphic Softmax Evaluation}
\label{sec:def_and_challenges}

We first recall the standard softmax function in the plaintext setting.
For an input vector $\mathbf{x} = (x_1, \cdots, x_n)$, the $i$-th component is defined as
\begin{equation}
    \text{softmax}(\mathbf{x})_i
    =
    \frac{\exp(x_i)}{\sum_{j=1}^n \exp(x_j)}
    =
    \frac{\exp(x_i - x_\text{max})}{\sum_{j=1}^n \exp(x_j - x_\text{max})},
    \label{eq:stable_softmax}
\end{equation}
where $i \in \{1,\ldots,n\}$ and $x_\text{max} = \max_{i \in \{1,\ldots,n\}} x_i$.
In standard floating-point implementations, the second form is typically used as a numerically stable formulation to prevent overflow from the exponential function.
Evaluating softmax directly in the HE setting presents two major challenges:

\begin{enumerate}
    \item \textbf{Overflow and Maximum Subtraction:} 
    The rapid growth of the exponential function renders softmax vulnerable to numerical overflow when processing large input values.
    In the homomorphic setting, larger ranges necessitate polynomial approximations over wider intervals, further increasing computational cost and exacerbating approximation error~\cite{Cheon2022efficient,Cho2024fast}.
    To ensure numerical stability, the input vector to softmax is typically shifted using the maximum operation, as shown in Eq.~\eqref{eq:stable_softmax}.
    However, the maximum operation is not natively supported by HE schemes and must be approximated, which incurs substantial multiplicative depth.
    
    \item \textbf{Division Operation:} 
    Softmax inherently involves a division operation for normalization.
    Since HE schemes do not natively support division, prior work~\cite{Hong2022Secure,Lim2025tricycle,DeCastro2025encryptedllm} typically addresses this challenge by employing polynomial approximations of the reciprocal function $1/x$, such as Chebyshev polynomial approximation~\cite{Cheney1966introduction} or Goldschmidt's iterative algorithm~\cite{Goldschmidt1964applications}, leading to increased multiplicative depth.
\end{enumerate}

\subsection{Prior Work}

In the strictly non-interactive setting, prior work can be broadly categorized into two classes: softmax approximation and softmax replacement.
While the former directly approximates the softmax function using polynomials, the latter replaces softmax with simpler surrogate functions that are also realized through polynomial approximations.

\paragraph{Softmax Approximation.}

Prior work on \emph{softmax approximation} aims to approximate the original softmax function using polynomial functions. 
Following an early work on homomorphic softmax approximation~\cite{Hong2022Secure}, subsequent works focus on mitigating exponential overflow. 

The first strategy stabilizes computation by subtracting the maximum value $x_{\max}$.
HETAL~\cite{Lee2023hetal} computes the exact maximum; however, this incurs high computational cost due to comparison operations, which should be approximated by polynomials. 
Subsequent frameworks therefore adopt heuristic or statistical estimates to approximate $x_{\max}$ more efficiently.
EncryptedLLM~\cite{DeCastro2025encryptedllm} uses a fixed empirical maximum derived from training data.
Recognizing the instability of fixed values, Tricycle~\cite{Lim2025tricycle} and ARION~\cite{Yang2025arion} estimate an upper bound on the \emph{expected} maximum based on statistical assumptions.

The second strategy addresses overflow by compressing the input interval.
\citet{Cho2024fast} introduce the \emph{normalize-and-square} algorithm, which evaluates softmax on scaled-down inputs to ensure numerical stability and subsequently recovers the original outputs via the iterative identity:
\begin{equation}
\text{softmax}\left(\mathbf{x}/2^{k-1}\right)_i = \frac{\text{softmax}\left(\mathbf{x}/2^{k}\right)_i^2}{\sum_{j=1}^n \text{softmax}\left(\mathbf{x}/2^{k}\right)_j^2}.
\label{eq:ccs_identity}
\end{equation}
THOR~\cite{Moon2025thor} adopts this framework, differing primarily in the choice of approximation algorithms for the exponential function and division.

Crucially, regardless of the overflow mitigation strategy, all aforementioned methods rely on polynomial approximations for both the exponential function and the computationally expensive homomorphic division, resulting in increased multiplicative depth.

\begin{table}[t]
\caption{Comparison with prior softmax approximation methods. \cgfsoftmax eliminates both max and division operations while remaining adaptive to input statistics. (R: Required, NR: Not Required)} \label{tab:benefitofcgf}
\centering
    \resizebox{\textwidth}{!}{
    \begin{tabular}{cccccc}
         \toprule          
         \multirow{2}{*}{\textbf{Operations}} & \multirow{2}{*}{HETAL~\cite{Lee2023hetal}} & \multirow{2}{*}{EncryptedLLM~\cite{DeCastro2025encryptedllm}} & Tricycle~\cite{Lim2025tricycle} & \citet{Cho2024fast} & \multirow{2}{*}{\textbf{\cgfsoftmax}}\\
         & & & ARION~\cite{Yang2025arion} & THOR~\cite{Moon2025thor} &\\
         \midrule[0.8pt]
         Max & Comparison-Based Approx. & Fixed (Empirical) & Expected Max Upper Bound & Fixed (Domain Scaling) &  \textbf{NR} \\
         \cmidrule{1-6}
         Division & R & R & R & R & \textbf{NR} \\
         \cmidrule{1-6}
         Data-Adaptivity & \textbf{Yes} & No & \textbf{Yes} & No &  \textbf{Yes} \\
         \bottomrule
    \end{tabular}
    }
\end{table}

\paragraph{Softmax Replacement.}

Distinct from softmax approximation, \emph{softmax replacement} methods substitute the computationally expensive softmax operation with simpler surrogate functions.
Several works replace softmax with element-wise activation functions, including standard activations (e.g., ReLU and GELU), their polynomial variants (e.g., squared ReLU), or activations augmented with auxiliary neural networks~\cite{Zimerman2024converting,Chen2022the}.
\citet{Rho2025encryption} pursue an alternative direction by adopting Gaussian kernels in place of softmax.
Another line of work replaces only the exponential term $\exp(x_i)$ with power functions or their shifted variants~\cite{Li2023mpcformer,Luo2024secformer,Zimerman2024power}, while retaining the division operation for normalization, e.g., Power-softmax~\cite{Zimerman2024power}, $x_i^{p}/{\sum_{j} x_j^{p}}$ with a positive even $p$.
BPMax from Powerformer~\cite{Park2025Powerformer} further eliminates the division operation by replacing the denominator with a fixed empirical value, e.g., $(x_i + c)^p$ with a positive odd $p$.

However, these softmax replacement methods face inherent limitations: they either suffer from accuracy degradation~\cite{Zimerman2024converting,Rho2025encryption} or provide limited speedup due to the retained division operation~\cite{Li2023mpcformer,Luo2024secformer,Zimerman2024power}.
Furthermore, even recent replacement methods~\cite{Park2025Powerformer,Luo2024secformer,Rho2025encryption} are primarily evaluated on small-scale models such as BERT~\cite{Devlin2019bert} and benchmarks with small label spaces (e.g., binary or ternary classification).
The effectiveness of such methods on larger models and more challenging benchmarks remains underexplored. 

\section{Softmax Reformulation via Cumulant Generating Function} \label{sec:CGF-softmax}

In this section, we introduce a novel reformulation of the softmax function based on the CGF that enables efficient evaluation under HE.

\subsection{CGF-softmax}

\paragraph{CGF-softmax.}

We introduce CGF-softmax to address the challenges of softmax in HE described in Section~\ref{sec:def_and_challenges}.
Given an input vector $\mathbf{x} = (x_1, \cdots, x_n)$, we interpret the softmax denominator as a scaled sample mean of exponential terms.
Let $X$ denote a random variable representing the distribution of input components, and let $M_X(t) \triangleq \mathbb{E}[\exp(tX)]$ denote its MGF.
Under this interpretation, the softmax denominator $\sum_{i=1}^{n} \exp(x_i)$ is replaced by $n M_X(1)$ by substituting the sample mean with the ensemble mean.
Accordingly, the $i$-th component of \cgfsoftmax is given by
\begin{equation}
\text{softmax}_{\text{CGF}}(\mathbf{x})_{i} = \frac{\exp(x_i)}{n M_X (1)}. 
\label{eq:mgf-softmax}
\end{equation}

To make the structure explicit, we introduce the cumulant generating function (CGF) of $X$~\cite{Casella2024Statistical}, defined as $K_X(t) \triangleq \ln M_X(t)$.\label{eq:cgf}
Using the CGF, Eq.~\eqref{eq:mgf-softmax} is expressed in an explicit single-exponential form as
\begin{align}
\text{softmax}_{\text{CGF}} (\mathbf{x})_{i}
& = \exp \left (x_i - K_X(1) - \ln n \right ).  
\label{eq:cgf-softmax}
\end{align}
Moreover, $K_X(1)$ can be expanded as
\begin{equation} \label{eq:cgf-cumulants}
K_X(1) = \sum_{j=1}^{\infty} \frac{\kappa_j}{j!},
\end{equation}
where $\kappa_j$ denotes the $j$-th \emph{cumulant} of $X$. 
This representation is particularly suitable for HE since cumulants, such as $\kappa_1=\mu$ (mean) and $\kappa_2=\sigma^2$ (variance) can be efficiently computed using basic HE operations.

\paragraph{Advantages of CGF-softmax.}

By reformulating softmax via CGF, \cgfsoftmax achieves both computational efficiency and improved approximation stability via a data-adaptive shift.
This reformulation leads to the following key advantages:
\begin{enumerate}
    \item \cgfsoftmax eliminates the need for a $\max$ operation. 
    All input components are shifted by $K_X(1) + \ln n$ prior to exponential evaluation, which preserves numerical stability.
    \item \cgfsoftmax removes the division operation and correspondingly reduces the multiplicative depth, as the normalization term is absorbed into the exponent of the single exponential. 
    \item \cgfsoftmax is adaptive to the input statistics through cumulants. 
    This data-adaptivity leads to more stable approximation behavior across varying input samples.
\end{enumerate} 
As a result, \cgfsoftmax directly addresses the challenges of softmax evaluation in HE outlined in Section~\ref{sec:def_and_challenges}. 
Table~\ref{tab:benefitofcgf} compares \cgfsoftmax with prior softmax approximation methods and summarizes the advantages of the proposed approach.
In particular, it shows that \cgfsoftmax uniquely eliminates both $\mathrm{max}$ and division operations while remaining adaptive to input statistics.

\subsection{Properties of CGF-softmax} \label{subsec:properties}

In this subsection, we present key properties of \cgfsoftmax that reveal its advantages under HE. 

\paragraph{(P1) Positivity.}

\cgfsoftmax is strictly \emph{positive} for every component by Eq.~\eqref{eq:cgf-softmax}.

This property is not automatic for power function-based softmax replacements. 
For example, BPMax~\cite{Park2025Powerformer} can preserve ordering when $p$ is a positive odd integer, but if the shift $c$ is not large enough to ensure $x_i+c>0$ for all $i$, it may produce nonpositive outputs.

\paragraph{(P2) Order and Pairwise-Ratio Preservation.}

A softmax replacement should preserve the \emph{relative importance} of the input logits.
CGF-softmax satisfies this requirement by preserving both logit ordering and the exact pairwise output ratios of the standard softmax.

\begin{proposition}
\cgfsoftmax preserves the relative ordering of the input logits and has the same pairwise output ratios as the standard softmax. Specifically, for all $i,j\in\{1,\ldots,n\}$,
\begin{equation}
x_i>x_j
\;\Rightarrow\;
\text{softmax}_{\text{CGF}}(\mathbf{x})_i
>
\text{softmax}_{\text{CGF}}(\mathbf{x})_j,
\end{equation}
and
\begin{equation}
\frac{\text{softmax}(\mathbf{x})_i}
     {\text{softmax}(\mathbf{x})_j}
=
\frac{\text{softmax}_{\text{CGF}}(\mathbf{x})_i}
     {\text{softmax}_{\text{CGF}}(\mathbf{x})_j}
=
\exp(x_i-x_j).
\end{equation}
\end{proposition}

These properties follow directly from Eq.~\eqref{eq:cgf-softmax}: the exponential numerator is strictly increasing, and the denominator $nM_X(1)$ is a positive constant common across all indices.

In contrast, power-function-based replacements preserve ordering only under suitable conditions and generally do not preserve the exact exponential pairwise ratios of the standard softmax. 
Power-softmax~\cite{Zimerman2024power} with even $p$ preserves ordering only over nonnegative logits, since $x\mapsto x^p$ is not strictly increasing on $\mathbb{R}$, and its ratios are polynomial, $(x_i/x_j)^p$, when defined. 
BPMax~\cite{Park2025Powerformer} preserves ordering when $p$ is a positive odd integer and the denominator is a positive constant common across indices, but its ratios are also polynomial, $\left((x_i+c)/(x_j+c)\right)^p$, rather than $\exp(x_i-x_j)$.

\paragraph{(P3) Shift Invariance.} 

The standard softmax is \emph{shift-invariant}, i.e., $\text{softmax}(\mathbf{x}) = \text{softmax}(\mathbf{x} - c)$ for any constant $c$. 
This property is commonly used to improve numerical stability (see Eq.~\eqref{eq:stable_softmax}).
\cgfsoftmax also preserves this shift-invariance property. 

\begin{proposition}    \label{prop:shiftinvar}
    \cgfsoftmax is shift-invariant. For any constant $c \in \mathbb{R}$, 
    \begin{equation}
    \text{softmax}_{\text{CGF}}(\mathbf{x}) = \text{softmax}_{\text{CGF}}(\mathbf{x}-c).
    \end{equation}     
\end{proposition}
\begin{proof}
The proof is given in Appendix~\ref{appendix:cgf_shiftinvariant}. 
\end{proof}
In contrast, softmax replacement methods~\cite{Luo2024secformer,Zimerman2024power,Park2025Powerformer} are generally not shift-invariant.

\paragraph{(P4) Domain Scaling.}

Efficient evaluation of the softmax function under HE is challenging when the input values span a wide range, as polynomial approximations of the exponential function rapidly deteriorate outside a limited interval~\cite{Cheon2022efficient}.
To address this issue, it is crucial to control the domain of the exponential function during evaluation.

A representative approach is the normalize-and-square strategy proposed by~\cite{Cho2024fast}, which scales the input domain to a small interval and subsequently restores the original scale via repeated squaring (see Eq.~\eqref{eq:ccs_identity}).
While effective, applying this strategy to the standard softmax requires repeated division operations, which are particularly costly under HE.

In contrast, \cgfsoftmax naturally enables a much simpler \emph{domain scaling} mechanism.
\begin{proposition}
The domain scaling of CGF-softmax by a factor of $2^k$ ($k \in \mathbb{N}$) is expressed as:
    \begin{equation} \label{eq:scaling}
        \text{softmax}_{\text{CGF}}(\mathbf{x})_{i} = \exp \left(\frac{x_i - K_X(1) - \ln n}{2^k} \right)^{2^k}. 
    \end{equation}
\end{proposition}
The exponent is scaled by a factor of $1/2^k$ to ensure that the exponential function is evaluated over a significantly smaller domain.
The original value is then recovered by squaring the result $k$ times.
Crucially, this domain scaling reduces the effective approximation interval of the exponential function by a factor of $2^k$, while incurring only an additional multiplicative depth of $k+1$ ($\mathsf{CMult}$: $k$, $\mathsf{PMult}$: $1$).
Unlike the standard softmax, this procedure does not involve explicit normalization or division, enabling a substantially simpler and more efficient realization under HE.

\subsection{Error Analysis}


To evaluate the approximation accuracy of \cgfsoftmax, we analyze the \emph{relative error} as in \cite{Cho2024fast}.
\begin{proposition}
    The relative error between CGF-softmax and softmax is given by
    \begin{align}
        \eta & = \frac{\left\|\text{softmax}(\mathbf{x})-\text{softmax}_{\text{CGF}}(\mathbf{x})\right\|_\infty}{\left\|\text{softmax}(\mathbf{x})\right\|_\infty} = \left|1-\frac{\frac{1}{n}\sum_{i=1}^n \exp(x_i)}{\mathbb{E}[\exp(X)]}\right|. \label{eq:rel_err}
    \end{align}
\begin{proof}
The proof is given in Appendix~\ref{appendix:rel_error}. 
\end{proof}

\end{proposition}
Importantly, the relative error $\eta$ depends only on the ratio between the sample mean and the ensemble mean of $\exp(X)$, and is independent of the individual component index $i$.

To further quantify the concentration of the relative error, we apply the Berry--Esseen theorem~\cite{Berry1941accuracy,Esseen1942liapounoff}.
\begin{proposition}
Let $X_1,\ldots,X_n$ be i.i.d. random variables with the same distribution as $X$, set $Y_i=\exp(X_i)$, and define $\mu_Y=\mathbb{E}[Y_i]$, $\sigma_Y^2=\text{Var}(Y_i)>0$, and $\rho_Y=\mathbb{E}\left[|Y_i-\mu_Y|^3\right]<\infty$.
Then, for any $\delta>0$,
\begin{equation}
P_\eta \triangleq P(\eta \ge \delta) \le 2\left[ 1-\Phi\left( \frac{\delta\mu_Y\sqrt{n}}{\sigma_Y} \right) \right] + \frac{C\rho_Y}{\sigma_Y^3\sqrt{n}},
\label{eq:be_rel_err_tail}
\end{equation}
where $C>0$ is a constant and $\Phi(\cdot)$ is the cumulative distribution function (CDF) of the standard normal distribution.
\end{proposition}
\begin{proof}
The proof is given in Appendix~\ref{appendix:proof_berry_esseen}.
\end{proof}
Since both the standard-normal tail term and the Berry--Esseen correction vanish as $n$ increases, $P_\eta\to 0$ for any $\delta>0$. 
Thus, the relative error $\eta$ asymptotically vanishes in probability.

\subsection{Practical Instantiation: Second-Order Cumulant Approximation}

While Eq.~\eqref{eq:cgf-softmax} expresses CGF-softmax in terms of the 
full CGF $K_X(1)$, practical evaluation under 
HE requires a tractable approximation. 
We adopt the second-order cumulant approximation, $K_X(1) \approx \mu + \sigma^2/2$. 
Here, $\mu$ and $\sigma^2$ correspond to the first two cumulants $\kappa_1$ and $\kappa_2$, while higher-order terms $\kappa_j/j!$ for $j \geq 3$ are omitted. 
This approximation requires estimating only $\mu$ and $\sigma^2$, which are efficiently 
computed under HE.
The omitted terms remain sufficiently small to have negligible impact on inference accuracy, supported by two complementary reasons. First, the factorial denominator $j!$ intrinsically suppresses higher-order contributions. Second, fine-tuning empirically suppresses the higher-order cumulants $\kappa_j$ themselves, as we demonstrate in Appendix~\ref{appx:higher_order}.

\section{Experimental Results} \label{sec:experimental}

\begin{table}[t]
    \centering
    \caption{Comparison with prior softmax approximation methods. 
    \cgfsoftmax achieves significantly lower multiplicative depth compared to existing baselines even for larger-scale models (e.g., LLaMA-3.2-1B) and more complex benchmarks.}
    \label{tab:depth_analysis}
    \resizebox{\linewidth}{!}{
    \begin{tabular}{llcccc}
        \toprule
        \textbf{Method} & \textbf{Model} & \textbf{Model Size} & \textbf{Dataset} & \textbf{\# Classes} & \textbf{Mult. Depth} \\
        \midrule
        \multirow{3}{*}{EncryptedLLM~\cite{DeCastro2025encryptedllm}} & GPT-2 Small & 124M & \multirow{3}{*}{\shortstack[l]{SST-2, WiC, PIQA, MNLI, ANLI, \\ Social IQA, HellaSwag, ARC (Easy) }} & \multirow{3}{*}{2, 3, 4} & 22 \\
                                      & GPT-2 Medium & 355M & & & 26 \\
                                      & GPT-2 Large & 774M & & & 30 \\
        \midrule
        THOR~\cite{Moon2025thor} & BERT-Base & 110M & SST-2, RTE, MRPC & 2 & 30 \\
        Tricycle~\cite{Lim2025tricycle} & BERT-Tiny & 4.4M & SST-2 & 2 & 17 \\
        
        ARION~\cite{Yang2025arion} & BERT-Tiny/Base & \shortstack{4.4M / 110M} & SST-2, RTE, QNLI & 2 & 14 \\
        \midrule
        \multirow{2}{*}{\textbf{Proposed}} & LLaMA-3.2-1B & 1B & \shortstack[l]{SST-2, Banking77, Clinc150} & \textbf{2, 77, 150} & \textbf{8--9} \\
        \cmidrule(l){2-6}
         & ViT/DeiT (Tiny/Base)& 5M / 86M & ImageNet-1k & \textbf{1,000} & \textbf{7--10} \\
        \bottomrule
    \end{tabular}
    }
\end{table}
We implemented our HE algorithms using \texttt{desilofhe}, a Python-based HE library developed by \cite{Desilo2025library}, which supports the CKKS scheme.
We set the slot size as $s = 2^{15}$, and the available multiplicative level after bootstrapping as $L=10$.
All experiments were conducted on dual AMD EPYC 7763 64-Core Processors (totaling 128 physical cores) and \SI{1}{\tera\byte} of RAM, operating on Ubuntu 20.04 LTS.

\subsection{Softmax Evaluation}

Table~\ref{tab:depth_analysis} presents a comparative analysis of the multiplicative depth required by CGF-softmax and prior softmax approximation methods~\cite{DeCastro2025encryptedllm,Lim2025tricycle,Yang2025arion}.
We note that prior approaches are optimized for fixed experimental settings tailored to specific model architectures and datasets.
As a result, a direct comparison under identical input intervals is inherently challenging due to differences in the evaluated models and datasets. 
Despite these discrepancies, the comparison in Table~\ref{tab:depth_analysis} highlights the \emph{scalability} and \emph{efficiency} of \cgfsoftmax. 
In particular, \cgfsoftmax achieves the lowest multiplicative depth, ranging from 7 to 10, even when applied to significantly larger-scale models (e.g., LLaMA-3.2-1B) and high-complexity benchmarks such as ImageNet-1k with 1,000 classes.
In contrast, prior methods require substantially higher multiplicative depths while being evaluated on smaller models and tasks with fewer classes.

Next, distinct from the softmax approximation methods discussed above, we evaluate the algorithmic complexity of CGF-softmax against~\citet{Cho2024fast}.
We adopt it as our primary softmax approximation baseline because, like our proposed method, it employs a \emph{domain scaling} strategy that enables a generalizable and scalable framework capable of handling varying input intervals, rather than relying on fixed bounds or dataset-specific tuning.
Table~\ref{tab:softmax_complexity} presents a detailed comparison.
To handle large input intervals $[-M, 0]$, \citet{Cho2024fast} mitigates numerical instability by scaling inputs by a factor of $1/2^k$.
However, recovering the original scale requires the \emph{normalize-and-square} algorithm (see Eq.~\eqref{eq:ccs_identity}), which involves an iterative process that is computationally expensive. 
In practice, this typically requires $k \approx \lceil \log_2{M}-\log_2(\ln{n})\rceil$ iterations, where $n$ denotes the input dimension of softmax.
In contrast, \cgfsoftmax streamlines this process by directly exploiting Eq.~\eqref{eq:scaling}.
This enables efficient handling of large input intervals while significantly reducing computational overhead. 
These advantages are reflected in the reduced multiplicative depth reported in Table~\ref{tab:softmax_complexity}.
Implementation details of \cgfsoftmax are provided in Appendix~\ref{appx:implementation}.

We also measure the execution time of softmax evaluation under HE. 
For this experiment, we set the input size to a $256 \times 256$ matrix with an input interval of $[-128,0]$, and all timings are measured using 8 CPU threads.
As shown in Table~\ref{tab:softmax_runtime}, the observed latency reduction is primarily attributed to the reduced multiplicative depth, which reduces or eliminates the need for expensive bootstrapping, as well as the reduced number of polynomial approximation steps, which significantly lowers the number of $\mathsf{CMult}$ operations.

Additionally, since \cgfsoftmax eliminates homomorphic division, it requires fewer homomorphic operations than the standard softmax, resulting in lower evaluation noise.
In our empirical analysis, \cgfsoftmax achieves a total noise of $3.05 \times 10^{-10}$ at multiplicative depth 7, which is lower than that of the standard softmax formulation.
This lower-noise behavior is particularly beneficial for applying it to deeper HE neural networks, where preserving the noise budget is critical.
A detailed empirical comparison of the noise behavior is provided in Appendix~\ref{appx:noise_analysis}.

\begin{table}[t]
    \centering
    \caption{Algorithmic complexity comparison between \cgfsoftmax and~\citet{Cho2024fast}, both based on domain scaling by $1/2^k$. 
    Here, $n$ denotes the input dimension of softmax, and $L$ represents the number of available multiplicative levels after bootstrapping.}
    \label{tab:softmax_complexity}
    \resizebox{0.6 \linewidth}{!}{
    \begin{tabular}{lcccc}
        \toprule
        \textbf{Method} & \textbf{Depth} & \textbf{\# $\mathsf{CMult}$} & \textbf{\# $\mathsf{Rot}$} & \textbf{\# $\mathsf{Boot}$} \\
        \midrule
        \citet{Cho2024fast} & $\ge 8k+9$ & $\ge 12k+58$ & $2k\log_2{n}$ & $\ge \lfloor (8k+9)/L \rfloor$ \\
        \textbf{Proposed} & $k + 6$ & $k + 10$ & $2\log_2{n}$ & $\lfloor (k+6)/L \rfloor$ \\
        \bottomrule
    \end{tabular}
    }
\end{table}

\begin{table}[t!]
    \centering
    \caption{Breakdown of CPU runtime (in seconds) for a single softmax evaluation on a $256 \times 256$ input matrix. 
    The reduction in total latency is primarily attributed to the elimination of bootstrapping ($\mathsf{Boot}$) and the reduced number of ciphertext multiplications ($\mathsf{CMult}$).
    In our setting, the average runtime per operation for $\mathsf{Add}$, $\mathsf{PMult}$, $\mathsf{CMult}$, $\mathsf{Rot}$, and $\mathsf{Boot}$ is approximately $0.0014$, $0.0021$, $0.089$, $0.059$, and $14$, respectively.
    Values are reported as mean $\pm$ standard deviation over 100 runs.}
    \label{tab:softmax_runtime}
    \resizebox{0.75 \linewidth}{!}{
    \begin{tabular}{lcccccc}
    \toprule
    \textbf{Method} & $\mathsf{Add}$ & $\mathsf{PMult}$ & $\mathsf{CMult}$ & $\mathsf{Rot}$ & $\mathsf{Boot}$ & \textbf{Total} \\
    \midrule
    \citet{Cho2024fast} 
    & $2.46 \pm 0.08$ 
    & $1.20 \pm 0.04$ 
    & $11.65 \pm 0.26$ 
    & $1.48 \pm 0.06$ 
    & $90.86 \pm 2.99$ 
    & $107.65 \pm 3.16$ \\
    \textbf{Proposed} 
    & $0.32 \pm 0.02$ 
    & $0.12 \pm 0.01$ 
    & $\mathbf{1.95 \pm 0.10}$ 
    & $0.83 \pm 0.06$ 
    & $\mathbf{0.00 \pm 0.00}$ 
    & $\mathbf{3.25 \pm 0.15}$ \\
    \bottomrule
\end{tabular}
    }
\end{table}

\begin{figure}[t]
    \centering
    \includegraphics[width=0.8\linewidth]{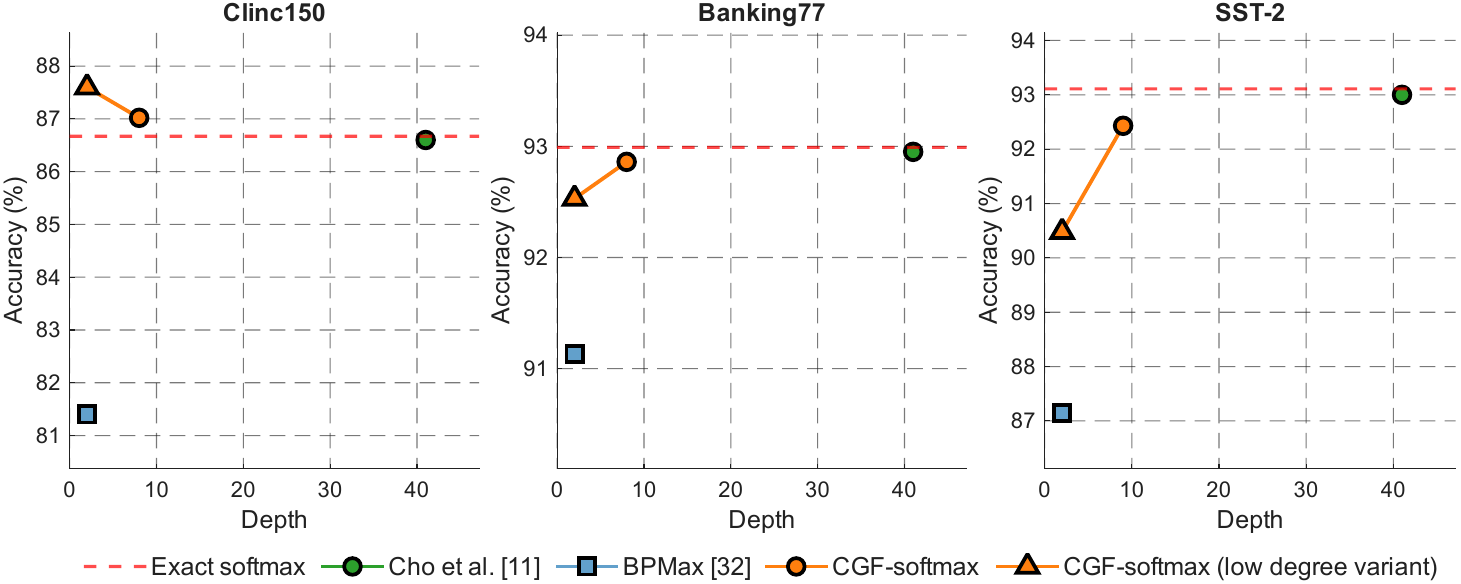}
    \caption{
    Cost–accuracy trade-off for LLaMA-3.2-1B across NLP benchmarks (Clinc150, Banking77, and SST-2). 
    CGF-softmax achieves near-plaintext accuracy with substantially lower multiplicative depth, thereby Pareto-dominating both \citet{Cho2024fast} (exact approximation) and BPMax~\cite{Park2025Powerformer} (low-depth replacement).
    }
    \label{fig:llama_depth}
\end{figure}

\begin{figure}[t]
    \centering
    \includegraphics[width=0.9\linewidth]{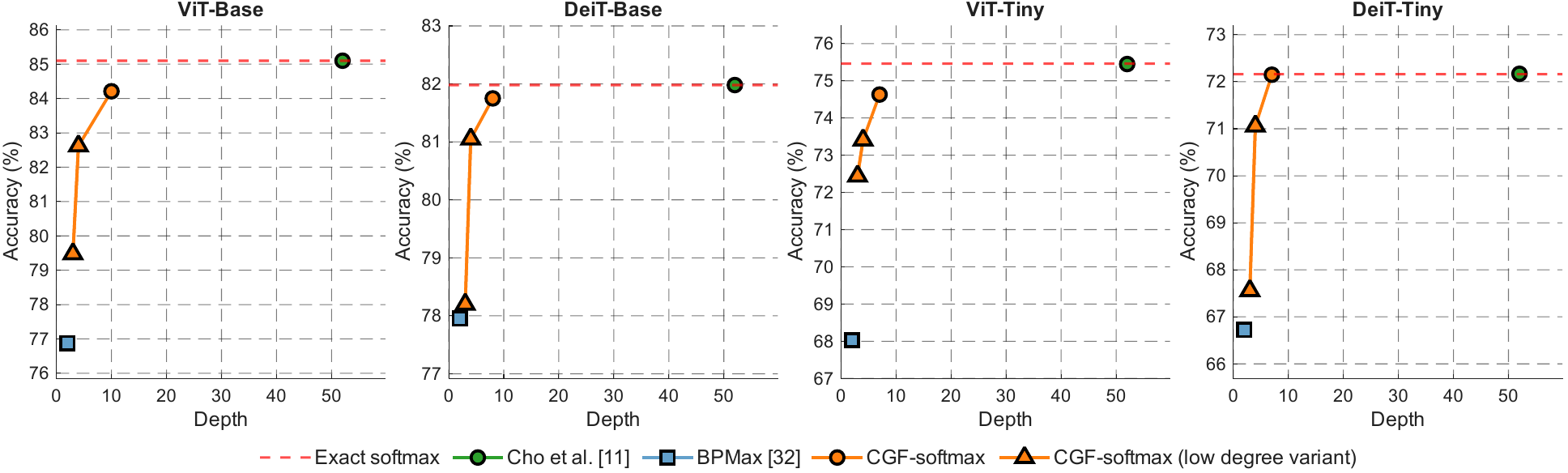}
    \caption{
    Cost–accuracy trade-off for ViT and DeiT models on the ImageNet-1k dataset. 
    CGF-softmax achieves near-plaintext accuracy with substantially lower multiplicative depth, thereby Pareto-dominating both \citet{Cho2024fast} (exact approximation) and BPMax~\cite{Park2025Powerformer} (low-depth replacement).
    }
    \label{fig:vit_depth}
\end{figure}
\subsection{Model Accuracy Evaluation}
\subsubsection{Experimental Setup}

\paragraph{Evaluation Scope \& Objectives.}

To validate the scalability of \cgfsoftmax to large-scale architectures and its ability to handle complex tasks with high-dimensional output spaces, we consider the following evaluation 
(1) LLaMA-3.2-1B, sourced from the Hugging Face \texttt{transformers} library~\cite{wolf2020transformers}, evaluated on Clinc150~\cite{Larson2019evaluation}, Banking77~\cite{Casanueva2020efficient}, and SST-2~\cite{Wang2018glue}, which involve $150$, $77$, and $2$ classes, respectively; and (2)
Vision Transformers (ViT/DeiT-Base/Tiny) implemented using the \texttt{timm} library~\cite{Wightman2019timm}, evaluated on ImageNet-1k~\cite{Dosovitskiy2021image} (1,000 classes).
\paragraph{Baselines.}

We select representative baselines from both the softmax approximation and replacement categories.
For \emph{softmax approximation}, we use~\citet{Cho2024fast} as the primary baseline.
Unlike prior approaches that are optimized for specific experimental setting, this method provides a generalized framework capable of handling arbitrary input intervals, making it suitable for evaluation across diverse benchmarks. 
For \emph{softmax replacement}, we use Batch Power-Max (BPMax) from Powerformer~\cite{Park2025Powerformer}, which eliminates division to achieve low circuit depth and was reported to incur negligible inference-accuracy degradation on BERT models.
\paragraph{Training Methodology.}

To recover the inference accuracy affected by homomorphic softmax evaluation, prior replacement baselines (BPMax) require an additional training phase.
For a fair comparison, we fine-tune \cgfsoftmax using the same training framework as BPMax.
Details of the training configurations and hyperparameters are provided in Appendix~\ref{appx:search_space}.

Notably, \cgfsoftmax converges faster during this adaptation process, leading to substantially lower training cost.
For LLaMA-3.2-1B, \cgfsoftmax reaches its best performance within $5$ epochs, whereas BPMax requires $20$ epochs.
For Vision Transformers, we further reduce the training cost by applying attention-only fine-tuning. We also present an \emph{alternative pipeline} that removes this additional training cost by incorporating \cgfsoftmax directly into the initial Supervised Fine-Tuning (SFT) phase.
Details are provided in Appendix~\ref{appx:alt_pipeline}.
\paragraph{Low-degree Variant.}
We further introduce a \emph{low-degree variant} of \cgfsoftmax, designed for scenarios that prioritize maximal efficiency with minimal multiplicative depth.
Unlike the standard setting, this variant modifies the approximation stage by replacing the exact exponential function with a low-degree polynomial during fine-tuning, rather than approximating a fine-tuned model at inference time. 
Details of the exponential approximation schemes are provided in Appendix~\ref{appx:exp_approx}.

\subsubsection{Inference Accuracy}
We analyze the inference accuracy results presented in Figures~\ref{fig:llama_depth} and~\ref{fig:vit_depth}.
In these figures, \emph{exact softmax} denotes the accuracy of the original pre-trained model using the standard softmax function in the plaintext domain.
The objective of HE-based softmax methods is to preserve the plaintext accuracy as closely as possible, while minimizing the required multiplicative depth.

Figures~\ref{fig:llama_depth} and~\ref{fig:vit_depth} collectively show the cost–accuracy Pareto frontier for homomorphic softmax evaluation. 
The two existing baselines occupy opposing extremes: \citet{Cho2024fast} attains near-exact accuracy but requires prohibitive multiplicative depths of up to $52$, while BPMax~\cite{Park2025Powerformer} achieves minimal depth of $2$ at the cost of substantial accuracy degradation. 
CGF-softmax Pareto-dominates both: it matches the accuracy of \citet{Cho2024fast} within \SI{1}{\percent} while requiring $4$–$6\times$ less depth, and it surpasses BPMax~\cite{Park2025Powerformer} by large accuracy margins at comparable or only slightly higher depth. 
Notably, even the low-degree variant of CGF-softmax--which operates at the same depth as BPMax~\cite{Park2025Powerformer} in the LLaMA setting--outperforms it across all benchmarks.
\section{Conclusion}
\label{sec:conclusion}
In this paper, we introduce CGF-softmax, a novel reformulation of the softmax function under HE. 
By leveraging the CGF, CGF-softmax reformulates the softmax denominator to eliminate computationally expensive operations such as maximum subtraction and homomorphic division. 
We provide a theoretical analysis establishing the key properties of CGF-softmax, including an asymptotic bound on its approximation error.
Extensive experiments on large-scale architectures, including LLaMA-3.2-1B and Vision Transformers, demonstrate that CGF-softmax preserves inference accuracy within \SI{1}{\%} of the plaintext baseline and consistently outperforms existing approaches on complex benchmarks.
By reconciling high accuracy with substantially reduced multiplicative depth, CGF-softmax offers a scalable and efficient solution for secure inference. 

\textbf{Limitations.} For relatively small models (e.g., BERT) on simple tasks (e.g., SST-2), existing softmax replacements (e.g., Power-softmax~\cite{Zimerman2024power}, BPMax~\cite{Park2025Powerformer}) already achieve satisfactory accuracy at low multiplicative depth, and thus may be sufficient in such settings.

{
\small
\bibliographystyle{abbrvnat}
\bibliography{abrv,mybib}

@STRING{IEEE_M_SP = "IEEE Security and Privacy"}

@STRING{IEEE_J_IFS = "IEEE Transactions on Information Forensics and Security"}

@STRING{SIAM_J_COMP = "SIAM Journal on Computing"}

@STRING{NAACL = "Proceedings of the Annual Conference of the North American Chapter of the Association for Computational Linguistics (NAACL)"}

@STRING{FnT_PRIV = "Foundations and Trends in Privacy and Security"}

@STRING{ARXIV = "arXiv preprint"}

@STRING{NIPS = "Advances in Neural Information Processing Systems (NeurIPS)"}

@STRING{CCS = "Proceedings of the ACM SIGSAC Conference on Computer and Communications Security (CCS)"}

@STRING{ASIACRYPT = "Proceedings of the International Conference on the Theory and Application of Cryptology and Information Security (ASIACRYPT)"}

@STRING{ICML = "Proceedings of the International Conference on Machine Learning (ICML)"}

@STRING{ACL  = "Proceedings of the Annual Meeting of the Association for Computational Linguistics (ACL)"}

@STRING{IACR_EPRINT = "Cryptology ePrint Archive"}

@STRING{EMNLP = "Proceedings of the Conference on Empirical Methods in Natural Language Processing (EMNLP)"}

@STRING{FINDINGS_ACL = "Findings of the Association for Computational Linguistics (Findings of ACL)"}

@STRING{CVPR = "Proceedings of the IEEE/CVF Conference on Computer Vision and Pattern Recognition (CVPR)"}

@STRING{ICLR = "Proceedings of the International Conference on Learning Representations (ICLR)"}

@STRING{STOC = "Proceedings of the ACM Symposium on Theory of Computing (STOC)"}

@ARTICLE{Al-Rubaie2019Privacy,
  title   = "Privacy-Preserving Machine Learning: Threats and Solutions",
  author  = "Al-Rubaie, M and Chang, J M",
  journal = IEEE_M_SP,
  volume  =  17,
  number  =  2,
  pages   = "49--58",
  year    =  2019
}

@article{Berry1941Accuracy,
  title={{The accuracy of the Gaussian approximation to the sum of independent variates.}},
  author={Berry, Andrew C},
  journal={Transactions of the American Mathematical Society},
  volume={49},
  number={1},
  pages={122--136},
  year={1941},
  publisher={JSTOR}
}

@inproceedings{Brown2020language,
author = {Tom Brown and Benjamin Mann and Nick Ryder and Melanie Subbiah and Jared D. Kaplan and Prafulla Dhariwal and Arvind Neelakantan and Pranav Shyam and Girish Sastry and Amanda Askell and others},
booktitle = NIPS,
title = {{Language models are few-shot learners}},
volume = {33},
pages = {1877--1901},
month = dec,
year = {2020}
}

@inproceedings{Casanueva2020efficient,
author = {I{\~n}igo Casanueva and Tadas Tem{\v{c}}inas and Daniela Gerz and Matthew Henderson and Ivan Vuli{\'{c}}},
booktitle = {Proceedings of the 2nd Workshop on NLP for Conversational AI},
pages = {38--45},
month = jul,
title = {{Efficient intent detection with dual sentence encoders}},
year = {2020}
}

@BOOK{Casella2024Statistical,
  title     = {{Statistical inference}},
  author    = "Casella, George and Berger, Roger",
  publisher = "Chapman and Hall/CRC",
  address   = "Boca Raton",
  month     =  apr,
  year      =  2024,
  language  = "en"
}

@inproceedings{Chen2022the,
author = {Tianyu Chen and Hangbo Bao and Shaohan Huang and Li Dong and Binxing Jiao and Daxin Jiang and Haoyi Zhou and Jianxin Li and Furu Wei},
booktitle = FINDINGS_ACL,
pages = {3510--3520},
month = may,
title = {{THE-X: Privacy-preserving transformer inference with homomorphic encryption}},
year = {2022}
}

@book{Cheney1966introduction,
author = {Elliott Ward Cheney},
address = {New York},
publisher = {McGraw-Hill},
title = {{Introduction to approximation theory}},
year = {1966}
}

@inproceedings{Cheon2017homomorphic,
author = {Jung Hee Cheon and Andrey Kim and Miran Kim and Yongsoo Song},
booktitle = ASIACRYPT,
pages = {409--437},
month = dec,
title = {{Homomorphic encryption for arithmetic of approximate numbers}},
year = {2017}
}

@article{Cheon2022efficient,
  title={Efficient homomorphic evaluation on large intervals},
  author={Cheon, Jung Hee and Kim, Wootae and Park, Jai Hyun},
  journal=IEEE_J_IFS,
  volume={17},
  pages={2553--2568},
  year={2022},
  publisher={IEEE}
}

@inproceedings{Cho2024fast,
author = {Wonhee Cho and Guillaume Hanrot and Taeseong Kim and Minje Park and Damien Stehl{\'e}},
booktitle = CCS,
pages = {4391--4404},
month = oct,
title = {{Fast and accurate homomorphic softmax evaluation}},
year = {2024}
}

@inproceedings{DeCastro2025encryptedllm,
  author = {Leo De Castro and Daniel Escudero and Adya Agrawal and Antigoni Polychroniadou and Manuela Veloso},
  booktitle = ICML,
  title = {{EncryptedLLM: Privacy-preserving large language model inference via GPU-accelerated fully homomorphic encryption}},
  volume = {267},
  pages = {12677--12688},
  month = jul,
  year = {2025}
}

@inproceedings{Deng2009imagenet,
  author = {Jia Deng and Wei Dong and Richard Socher and Li-Jia Li and Kai Li and Li Fei-Fei},
  booktitle = CVPR,
  title = {{ImageNet: A large-scale hierarchical image database}},
  pages = {248--255},
  month = jun,
  year = {2009},
  doi = {10.1109/CVPR.2009.5206848}
}

@misc{Desilo2025library,
author = {{DESILO Inc.}},
title = {{DESILO FHE library}},
url = {https://desilo.ai},
year = {2025}
}

@inproceedings{Devlin2019bert,
author = {Jacob Devlin and Ming-Wei Chang and Kenton Lee and Kristina Toutanova},
booktitle = NAACL,
title = {{BERT}: Pre-training of deep bidirectional transformers for language understanding},
pages = {4171--4186},
month = jun,
year = {2019}
}

@inproceedings{Dosovitskiy2021image,
author = {Alexey Dosovitskiy and Lucas Beyer and Alexander Kolesnikov and Dirk Weissenborn and Xiaohua Zhai and Thomas Unterthiner and Mostafa Dehghani and Matthias Minderer and Georg Heigold and Sylvain Gelly and Jakob Uszkoreit and Neil Houlsby},
booktitle = ICLR,
title = {{An image is worth 16x16 words: Transformers for image recognition at scale}},
month = may,
year = {2021}
}

@article{Dubey2024llama,
author = {Abhimanyu Dubey and Abhinav Jauhri and Abhinav Pandey and Abhishek Kadian and Ahmad Al-Dahle and Aiesha Letman and Akhil Mathur and Alan Schelten and Alex Vaughan and Amy Yang and others},
journal = ARXIV,
title = {{The Llama 3 herd of models}},
volume = {arXiv:2407.21783},
month = jul,
year = {2024}
}

@book{Esseen1942liapounoff,
  title={{On the Liapounoff limit of error in the theory of probability}},
  author={Esseen, C.G.},
  series={Arkiv f{\"o}r matematik, astronomi och fysik},
  year={1942},
  publisher={Almqvist \& Wiksell}
}

@article{Evans2018pragmatic,
author = {David Evans and Vladimir Kolesnikov and Mike Rosulek},
journal = FnT_PRIV,
month = dec,
number = {2-3},
pages = {70--246},
title = {{A pragmatic introduction to secure multi-party computation}},
volume = {2},
year = {2018}
}

@inproceedings{Gentry2009fully,
author = {Craig Gentry},
booktitle = STOC,
pages = {169--178},
title = {{Fully homomorphic encryption using ideal lattices}},
month = may,
year = {2009}
}

@mastersthesis{Goldschmidt1964applications,
author = {Robert E. Goldschmidt},
school = {Massachusetts Institute of Technology},
title = {{Applications of division by convergence}},
month = jun,
year = {1964}
}

@ARTICLE{Hong2022Secure,
  title    = "Secure tumor classification by shallow neural network using
              homomorphic encryption",
  author   = "Hong, Seungwan and Park, Jai Hyun and Cho, Wonhee and Choe,
              Hyeongmin and Cheon, Jung Hee",
  journal  = "BMC Genomics",
  volume   =  23,
  number   =  1,
  pages    =  284,
  month    =  apr,
  year     =  2022
}

@inproceedings{Hu2022lora,
author = {Edward J. Hu and Yelong Shen and Phillip Wallis and Zeyuan Allen-Zhu and Yuanzhi Li and Shean Wang and Lu Wang and Weizhu Chen},
booktitle = ICLR,
title = {{{LoRA}: Low-rank adaptation of large language models}},
month = apr,
year = {2022}
}

@inproceedings{Larson2019evaluation,
author = {Stefan Larson and Anish Mahendran and Joseph J. Peper and Christopher Clarke and Andrew Lee and Parker Hill and Jonathan K. Kummerfeld and Kevin Leach and Michael A. Laurenzano and Lingjia Tang and Jason Mars},
booktitle = EMNLP,
pages = {1311--1316},
month = nov,
title = {{An evaluation dataset for intent classification and out-of-scope prediction}},
year = {2019}
}

@inproceedings{Lee2023hetal,
author = {Seewoo Lee and Garam Lee and Jung Woo Kim and Junbum Shin and Mun-Kyu Lee},
booktitle = ICML,
pages = {19010--19035},
title = {{HETAL: Efficient privacy-preserving transfer learning with homomorphic encryption}},
volume = {202},
month = jul,
year = {2023}
}

@inproceedings{Li2023mpcformer,
author = {Dacheng Li and Hongyi Wang and Rulin Shao and Han Guo and Eric Xing and Hao Zhang},
booktitle = ICLR,
title = {{MPCFormer: Fast, performant and private transformer inference with MPC}},
month = feb,
year = {2023}
}

@article{Lim2025tricycle,
author = {Lawrence Lim and Vikas Kalagi and Divyakant Agrawal and Amr El Abbadi},
journal = IACR_EPRINT,
title = {{Tricycle: Private transformer inference with tricyclic encodings}},
month = jun,
year = {2025}
}

@inproceedings{Lou2020safenet,
author = {Qian Lou and Yilin Shen and Hongxia Jin and Lei Jiang},
booktitle = ICLR,
title = {{SafeNet: A secure, accurate and fast neural network inference}},
month = jan,
year = {2020}
}

@inproceedings{Lou2021hemet,
author = {Qian Lou and Lei Jiang},
booktitle = ICML,
pages = {7102--7110},
title = {{HEMET}: A homomorphic-encryption-friendly privacy-preserving mobile neural network architecture},
month = jul,
year = {2021}
}

@inproceedings{Luo2024secformer,
author = {Jinglong Luo and Yehong Zhang and Zhuo Zhang and Jiaqi Zhang and Xin Mu and Hui Wang and Yue Yu and Zenglin Xu},
booktitle = FINDINGS_ACL,
pages = {13333--13348},
month = aug,
title = {{SecFormer: Fast and accurate privacy-preserving inference for transformer models via SMPC}},
year = {2024}
}

@misc{Meta2024meta,
  author = {{Meta}},
  title = {{meta-llama/Llama-3.2-1B: Model card}},
  url = {https://huggingface.co/meta-llama/Llama-3.2-1B},
  year = {2024}
}

@inproceedings{Moon2025thor,
author = {Jungho Moon and Dongwoo Yoo and Xiaoqian Jiang and Miran Kim},
booktitle = CCS,
pages = {3765--3779},
month = oct,
title = {{THOR: Secure transformer inference with homomorphic encryption}},
year = {2025}
}

@inproceedings{Park2025Powerformer,
author = {Dongjin Park and Eunsang Lee and Joon-Woo Lee},
booktitle = ACL,
pages = {11090--11111},
month = jul,
title = {{Powerformer: Efficient and high-accuracy privacy-preserving language model with homomorphic encryption}},
year = {2025}
}

@article{Paterson1973on,
author = {Mike Paterson and Larry J. Stockmeyer},
journal = SIAM_J_COMP,
month = mar,
number = {1},
pages = {60--66},
title = {{On the number of nonscalar multiplications necessary to evaluate polynomials}},
volume = {2},
year = {1973}
}

@inproceedings{Rho2025encryption,
author = {Donghwan Rho and Taeseong Kim and Minje Park and Jung Woo Kim and Hyunsik Chae and Ernest K. Ryu and Jung Hee Cheon},
booktitle = ICLR,
title = {{Encryption-friendly {LLM} architecture}},
month = jan,
year = {2025}
}

@article{Riazi2019deep,
title     = "Deep learning on private data",
author    = "Riazi, M Sadegh and Darvish Rouani, Bita and Koushanfar, Farinaz",
journal   = IEEE_M_SP,
publisher = "Institute of Electrical and Electronics Engineers (IEEE)",
volume    =  17,
number    =  6,
pages     = "54--63",
month     =  nov,
year      =  2019
}

@ARTICLE{Rivest1978On,
title   = "On data banks and privacy homomorphisms",
author  = "Rivest, R L and Adleman, L and Dertouzos, M L",
journal = {Foundations of Secure Computation},
volume  =  4,
number  =  11,
pages   = "169--180",
year    =  1978
}

@inproceedings{Touvron2021training,
  author = {Hugo Touvron and Matthieu Cord and Matthijs Douze and Francisco Massa and Alexandre Sablayrolles and Herv{\'e} J{\'e}gou},
  booktitle = ICML,
  title = {{Training data-efficient image transformers \& distillation through attention}},
  volume = {139},
  pages = {10347--10357},
  month = jul,
  year = {2021}
}

@inproceedings{Vaswani2017attention,
author = {Ashish Vaswani and Noam Shazeer and Niki Parmar and Jakob Uszkoreit and Llion Jones and Aidan N. Gomez and {\L}ukasz Kaiser and Illia Polosukhin},
booktitle = NIPS,
pages = {5998--6008},
title = {{Attention is all you need}},
month = dec,
year = {2017}
}

@inproceedings{Wang2018glue,
author = {Alex Wang and Amanpreet Singh and Julian Michael and Felix Hill and Omer Levy and Samuel R. Bowman},
booktitle = {Proceedings of the EMNLP Workshop BlackboxNLP},
title = {{GLUE: A multi-task benchmark and analysis platform for natural language understanding}},
pages = {353--355},
month = nov,
year = {2018}
}

@misc{Wightman2019timm,
author = {Ross Wightman},
title = {{timm: PyTorch image models}},
url = {https://github.com/rwightman/pytorch-image-models},
year = {2019}
}

@inproceedings{Wolf2020transformers,
author = {Thomas Wolf and Lysandre Debut and Victor Sanh and Julien Chaumond and Clement Delangue and Anthony Moi and Pierric Cistac and Tim Rault and Rémi Louf and Morgan Funtowicz and Jamie Brew and Hugues Bennett},
booktitle = {Proceedings of the EMNLP: System Demonstrations},
title = {{Transformers: State-of-the-art natural language processing}},
pages = {38--45},
month = nov,
year = {2020}
}

@article{Yang2025arion,
author = {Linhan Yang and Jingwei Chen and Wangchen Dai and Shuai Wang and Wenyuan Wu and Yong Feng},
journal = IACR_EPRINT,
title = {{ARION: Attention-optimized transformer inference on encrypted data}},
month = dec,
year = {2025}
}

@inproceedings{Zimerman2024converting,
  author = {Itamar Zimerman and Moran Baruch and Nir Drucker and Gilad Ezov and Omri Soceanu and Lior Wolf},
  booktitle = ICML,
  title = {{Converting transformers to polynomial form for secure inference over homomorphic encryption}},
  volume = {235},
  pages = {62803--62814},
  month = jul,
  year = {2024}
}

@article{Zimerman2024power,
author = {Itamar Zimerman and Allon Adir and Ehud Aharoni and Matan Avitan and Moran Baruch and Nir Drucker and Jenny Lerner and Ramy Masalha and Reut Meiri and Omri Soceanu},
journal = ARXIV,
title = {{Power-softmax: Towards secure LLM inference over encrypted data}},
volume = {arXiv:2410.09457},
month = oct,
year = {2024}
}
}


\appendix

\section{Shift Invariance of CGF-softmax}
\label{appendix:cgf_shiftinvariant}

For any constant $c$,
\begin{align}
\text{softmax}_{\text{CGF}}(\mathbf{x}-c)_i 
& = \frac{\exp(x_i - c)}{n \mathbb{E}[\exp(X-c)]} 
= \frac{\exp(x_i)}{n \mathbb{E}[\exp(X)]} = \text{softmax}_{\text{CGF}}(\mathbf{x})_i, \nonumber
\end{align}
which holds for all $i \in \{1, \ldots, n\}$.
Thus, CGF-softmax is shift-invariant.

\section{Relative Error} \label{appendix:rel_error}

The relative error $\eta$ is given by 
\begin{align}
\eta & =\frac{\left \| \text{softmax}(\mathbf{x}) - \text{softmax}_{\text{CGF}}(\mathbf{x}) \right \|_\infty}{\left \| \text{softmax}(\mathbf{x}) \right \|_\infty} \\
& = \frac{\left | \frac{1}{\sum_{i=1}^n \exp(x_i)} - \frac{1}{n M_X (1)} \right | \cdot \left \| \exp(\mathbf{x}) \right  \|_\infty}{\left | \frac{1}{\sum_{i=1}^n \exp(x_i)} \right |  \cdot \left \| \exp(\mathbf{x}) \right \|_\infty} \\
& = \left | 1 - \frac{\frac{1}{n} \sum_{i=1}^n \exp(x_i)}{ \mathbb{E}[\exp(X)]} \right |.
\end{align}
Note that the relative error depends only on the ratio between the sample mean and the ensemble mean. 

\section{Error Analysis using Berry--Esseen Theorem}
\label{appendix:proof_berry_esseen}
Let $\bar{Y}_n = \frac{1}{n}\sum_{i=1}^{n}Y_i$, where $Y_i=\exp(X_i)$.
By Eq.~\eqref{eq:rel_err}, the relative error can be written as
\begin{equation}
\eta = \left| 1-\frac{\bar{Y}_n}{\mu_Y} \right| = \frac{|\bar{Y}_n-\mu_Y|}{\mu_Y}.
\end{equation}
Define the standardized sample mean
\begin{equation}
Z_n = \frac{\sqrt{n}(\bar{Y}_n-\mu_Y)}{\sigma_Y}.
\end{equation}
Then, for any $\delta>0$,
\begin{equation}
\eta\ge\delta \quad\Longleftrightarrow\quad |Z_n| \ge \frac{\delta\mu_Y\sqrt{n}}{\sigma_Y}.
\end{equation}
Let
\begin{equation}
a_\delta = \frac{\delta\mu_Y\sqrt{n}}{\sigma_Y}.
\end{equation}
Then
\begin{equation}
P_\eta \triangleq P(\eta \ge \delta) = P(|Z_n|\ge a_\delta).
\end{equation}

By the Berry--Esseen theorem~\cite{Berry1941accuracy,Esseen1942liapounoff},
\begin{equation}
\sup_{z\in\mathbb{R}} \left| P(Z_n\le z)-\Phi(z) \right| \le \frac{C_{\text{BE}} \cdot \rho_Y}{\sigma_Y^3\sqrt{n}}.
\end{equation}
Therefore,
\begin{align}
P(|Z_n|\ge a_\delta)
&=P(Z_n\le -a_\delta) + P(Z_n\ge a_\delta) \\
&\le \Phi(-a_\delta) + \left[1-\Phi(a_\delta)\right] + 2 \cdot \frac{C_{\text{BE}} \cdot \rho_Y}{\sigma_Y^3\sqrt{n}}.
\end{align}
Since $\Phi(-a_\delta)=1-\Phi(a_\delta)$, we obtain
\begin{equation}
P_\eta \le 2\left[ 1-\Phi(a_\delta) \right] + \frac{2C_\text{BE} \cdot \rho_Y}{\sigma_Y^3\sqrt{n}}.
\end{equation}
Absorbing the factor $2$ into the universal constant $C ~(=2C_\text{BE})$ and substituting the definition of $a_\delta$ gives
\begin{equation}
P_\eta \le 2\left[ 1-\Phi\left( \frac{\delta\mu_Y\sqrt{n}}{\sigma_Y} \right) \right] + \frac{C \rho_Y}{\sigma_Y^3\sqrt{n}},
\end{equation}
which proves the claim.

\section{Empirical Behavior of Higher-order Terms in the CGF}
\label{appx:higher_order}

This section examines the effect of fine-tuning on the higher-order terms in the CGF.
Recall that
\begin{equation}
K_X(1)
=
\kappa_1 + \frac{\kappa_2}{2!}
+
\sum_{j=3}^{\infty} \frac{\kappa_j}{j!}.
\end{equation}
In \cgfsoftmax, we use the second-order cumulant approximation,
$K_X(1) \approx \kappa_1 + \kappa_2/2$, and omit the higher-order terms $\kappa_j/j!$ for $j \geq 3$.

Because the factorial denominator $j!$ suppresses higher-order terms as $j$ increases, contributions from order five and above are relatively small.
We therefore focus on the dominant omitted contributions, namely $\kappa_3/3!$ and $\kappa_4/4!$.

\begin{figure}[b]
    \centering
    \includegraphics[width=0.99\linewidth]{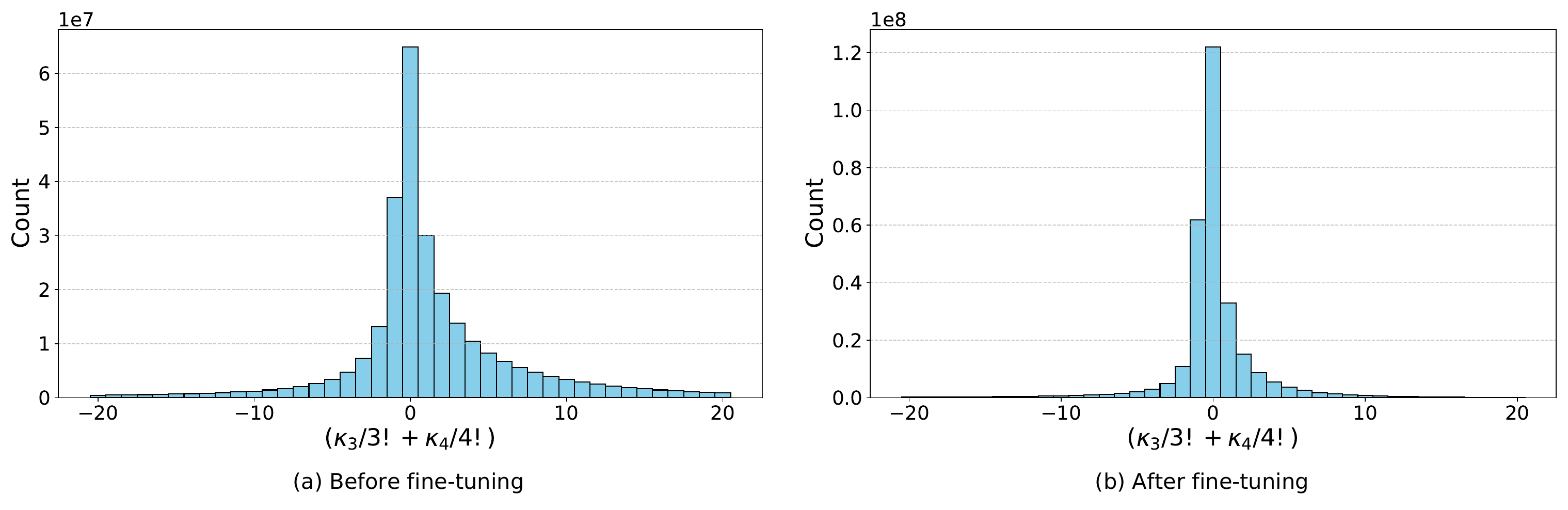}
    \caption{Distribution of $\kappa_3/3! + \kappa_4/4!$ for softmax inputs of DeiT-Base on ImageNet-1k, before and after fine-tuning.}
    \label{fig:kappa3_4}
\end{figure}

Figure~\ref{fig:kappa3_4} reports the distribution of $\kappa_3/3! + \kappa_4/4!$ computed from DeiT-Base softmax inputs on ImageNet-1k, before and after fine-tuning.
After fine-tuning, the distribution becomes more concentrated around zero, with its standard deviation reduced from $6.81$ to $3.85$.
This indicates that fine-tuning suppresses the dominant higher-order terms.

Together with the factorial attenuation inherent in the CGF, this empirical reduction supports the second-order cumulant approximation used in \cgfsoftmax and explains why the omitted higher-order terms have negligible impact on model performance in practice.

\section{Implementation Details} \label{appx:implementation}

\subsection{Packing Strategy} \label{appx:packing}
In Transformer architectures, the softmax function is predominantly applied within the self-attention mechanism, necessitating row-wise execution across the input matrix.
Accordingly, we target the row-wise application of CGF-softmax on an input matrix $A \in \mathbb{R}^{N_1 \times N_2}$.
We assume $N_1$ and $N_2$ are powers of two; for arbitrary dimensions, we apply zero-padding to extend them to the nearest power of two.
The input data is encrypted into $t = \lceil (N_1 \cdot N_2) / s \rceil$ ciphertexts.
To facilitate efficient column-wise aggregation, we employ a structure similar to the column-major packing variant used by Cho et al.~\cite{Cho2024fast}.
Specifically, elements within the same row are spaced by a fixed gap $g$, defined as $g = (t \cdot s) / N_2$.

Let $ct_i[j]$ denote the value stored in the $j$-th slot of the $i$-th ciphertext.
The element $A[n_1,n_2]$ located at row $n_1$ ($0 \le n_1 < N_1$) and column $n_2$ ($0 \le n_2 < N_2$) is mapped to the ciphertext as follows:
\begin{equation}
    A[n_1,n_2] = ct_i[j], \quad \text{where } 
    \begin{cases} 
        i = \lfloor (n_2 \cdot g) / s \rfloor, \\ 
        j = n_1 + (n_2 \cdot g \bmod s). 
    \end{cases}
    \label{eq:packing_map}
\end{equation}
Any slots in the ciphertexts that are not populated by matrix elements according to this mapping are initialized to zero.

\subsection{Algorithm}

Based on the packing strategy described in Appendix~\ref{appx:packing}, we present the outline of homomorphic CGF-softmax in Algorithm~\ref{alg:cgf_softmax}.
The input consists of a set of ciphertexts $\mathcal{V} = \{ct_0, \dots, ct_{t-1}\}$ representing the encrypted matrix.
The algorithm is designed to evaluate CGF-softmax in a row-wise manner.
Notably, the entire pre-processing stage (Steps 1 and 2) consumes 1 multiplicative level.
This is because the squaring operation necessitates a single ciphertext--ciphertext multiplication.
Furthermore, to optimize the circuit depth, consecutive scalar multiplications are pre-computed into a single scalar and applied to the ciphertext once; this is a common optimization technique in HE to prevent unnecessary level consumption.
For the exponential function, we employ a polynomial approximation denoted as $\mathsf{AExp}$, the details of which are outlined in Appendix~\ref{appx:exp_approx}.

\begin{algorithm}[H]
   \caption{Homomorphic CGF-softmax Outline}
   \label{alg:cgf_softmax}
   \centering
   
\begin{algorithmic}
   \STATE {\bfseries Input:} A set of ciphertexts $\mathcal{V} = \{ct_0, \dots, ct_{t-1}\}$ encrypting the matrix.
   \STATE {\bfseries Output:} A set of ciphertexts $\mathcal{W}= \{ct'_0, \dots, ct'_{t-1}\}$ encrypting the result of applying CGF-softmax row-wise to the matrix.

   \STATE \COMMENT{\textbf{// Step 1: Compute Mean ($\mu$)}}
   \STATE $ct_{\text{sum}} \leftarrow$ Sum all ciphertexts in $\mathcal{V}$ using addition.
   \STATE $\mu \leftarrow$ Perform row-wise summation on $ct_{\text{sum}}$ by repeating addition and rotation.

   \STATE \COMMENT{\textbf{// Step 2: Compute Variance Term ($\frac{1}{2}\sigma^2$)}}
   \STATE $\mathcal{V}_{\text{centered}} \leftarrow$ Subtract $\mu$ from each ciphertext in $\mathcal{V}$.
   \STATE $\mathcal{V}_{\text{sq}} \leftarrow$ Square each ciphertext in $\mathcal{V}_{\text{centered}}$.
   \STATE $ct_{\text{sum}} \leftarrow$ Sum all ciphertexts in $\mathcal{V}_{\text{sq}}$.
   \STATE $\text{var\_term} \leftarrow$ Perform row-wise summation on $ct_{\text{sum}}$.

   \STATE \COMMENT{\textbf{// Step 3: Final Approximation}}
   \STATE $\mathcal{V}_{\text{norm}} \leftarrow$ Subtract $\text{var\_term}$ and $\ln n$ from each ciphertext in $\mathcal{V}_{\text{centered}}$.
   \STATE $\mathcal{W} \leftarrow$ Apply $\mathsf{AExp}$ to each ciphertext in $\mathcal{V}_{\text{norm}}$.
   
   \STATE \textbf{return} $\mathcal{W}$
\end{algorithmic}
\end{algorithm}

\section{Noise Analysis under Homomorphic Evaluation}
\label{appx:noise_analysis}
In CKKS-based homomorphic evaluation, evaluation noise accumulates through arithmetic operations such as ciphertext multiplication, rescaling, and bootstrapping.
The amount of accumulated noise therefore depends on both the multiplicative depth and the number of homomorphic operations in the circuit.
From this perspective, homomorphic division is particularly costly, since it is implemented through an iterative approximation procedure involving additional multiplications and rescalings.
Thus, eliminating division directly reduces the main sources of noise accumulation in the softmax circuit.

To empirically show this, we compare the total noise induced when the standard softmax and \cgfsoftmax are evaluated homomorphically.
For an input vector $\mathbf{x}$, we denote the plaintext evaluations by
$\text{softmax}(\mathbf{x})$ and $\text{softmax}_{\text{CGF}}(\mathbf{x})$,
and their homomorphic evaluations by
$\text{softmax}_{\text{FHE}}(\mathbf{x})$ and
$\text{softmax}_{\text{CGF,FHE}}(\mathbf{x})$, respectively.

We measure the total noise by the $\ell_{\infty}$ error between the plaintext and homomorphic evaluation results:
\begin{equation}
\left\|
\text{softmax}(\mathbf{x}) -
\text{softmax}_{\text{FHE}}(\mathbf{x})
\right\|_{\infty},
\quad
\left\|
\text{softmax}_{\text{CGF}}(\mathbf{x}) -
\text{softmax}_{\text{CGF,FHE}}(\mathbf{x})
\right\|_{\infty}.
\end{equation}

For a fair comparison, both methods use the same degree-15 Chebyshev polynomial approximation for the exponential function, following the standard configuration in Appendix~\ref{appx:exp_approx}.
For \cgfsoftmax, the computation requires one depth for estimating the mean and variance before exponentiation,  four depths for the degree-15 exponential approximation, and two additional depths for the domain scaling (see Eq.~\eqref{eq:scaling}).
Therefore, the homomorphic evaluation of \cgfsoftmax is completed at multiplicative depth 7 without requiring homomorphic division.

The standard softmax uses the same exponential approximation, which accounts for multiplicative depth 6.
Unlike \cgfsoftmax, however, the standard softmax additionally requires homomorphic division for normalization.
While various methods exist for implementing division under HE, we adopt Goldschmidt's iterative algorithm~\cite{Goldschmidt1964applications}, since it is depth-efficient and requires relatively fewer homomorphic operations, thus introducing less encryption noise.
Each Goldschmidt iteration consumes one additional multiplicative depth and reduces the division error.

\begin{figure}[h]
    \centering
    \includegraphics[width=0.70\linewidth]{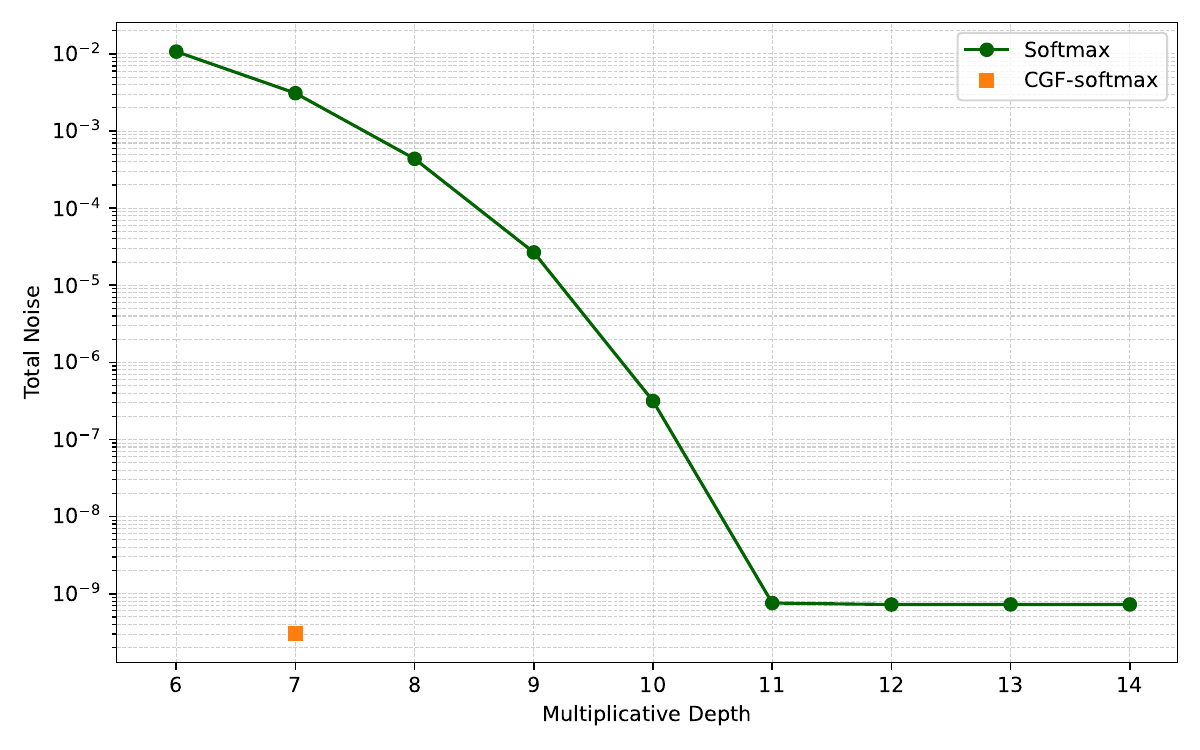}
    \caption{
    Total noise of the standard softmax and \cgfsoftmax under homomorphic evaluation, measured by the $\ell_{\infty}$ error between plaintext and homomorphic outputs.
    The $x$-axis denotes the overall multiplicative depth.
    For \cgfsoftmax, the computation is completed at depth 7 without homomorphic division, achieving a total noise of $3.05 \times 10^{-10}$.
    The standard softmax uses the same degree-15 exponential approximation and starts from depth 6, but additionally requires Goldschmidt's iterative division algorithm for normalization, where each iteration consumes one more multiplicative depth.
    Although increasing the number of division iterations reduces the total noise of the standard softmax, it remains noisier than \cgfsoftmax even at larger depths.
    }
    \label{fig:relative_error}
\end{figure}

Figure~\ref{fig:relative_error} reports the total noise as a function of the multiplicative depth.
For the standard softmax, increasing the number of Goldschmidt iterations reduces the total noise, but this improvement comes at the cost of additional multiplicative depth.
Moreover, even with larger depth, the total noise of the standard softmax remains higher than that of \cgfsoftmax.
In contrast, \cgfsoftmax avoids homomorphic division and achieves substantially lower total noise of $3.05 \times 10^{-10}$ at depth 7.
This error level is also small in absolute terms: it is about two orders of magnitude below the FP32 machine epsilon $(2^{-23} \approx 1.19 \times 10^{-7})$, and far below the precision levels of FP16 $(2^{-10} \approx 9.77 \times 10^{-4})$ and BF16 $(2^{-7} = 7.81 \times 10^{-3})$ arithmetic commonly used in neural network inference.
These results demonstrate that removing homomorphic division substantially reduces the noise introduced by softmax evaluation under HE, providing a clear advantage for noise-budget management in deeper HE neural networks.

\section{Approximation of the Exponential Function} \label{appx:exp_approx}
In this section, we detail the implementation of the polynomial approximation function $\mathsf{AExp}(x) \approx \exp(x)$, which constitutes the sole polynomial approximation step within the CGF-softmax framework.
The configuration of CGF-softmax--specifically, whether it operates as the \emph{standard configuration} or the efficient \emph{low-degree variant}--is determined by the timing of approximation relative to the training phase.

\subsection{Standard Configuration}
\label{appx:standard}
The standard CGF-softmax configuration applies exponential approximations in \emph{post-training} to prioritize maximum accuracy recovery.
In this setting, the model is trained using the exact CGF-softmax formulation, and the polynomial approximation is introduced only during the inference phase.
We utilize two distinct polynomial approximation methods for $\mathsf{AExp}(x)$.

\paragraph{Chebyshev Polynomial Approximation.}
To accommodate the wide dynamic range of input values while maintaining high precision, we leverage the exponential identity $\exp(x) = (\exp(x/2^k))^{2^k}$.
This formulation allows us to map the input to a reduced interval (e.g., $[-8, 0]$ following~\cite{Cho2024fast}) by scaling the domain by a factor of $1/2^k$.
Within this reduced interval, we employ a degree $d=15$ Chebyshev polynomial, with coefficients determined by a least-squares approximation over the target interval.
To evaluate this efficiently, we adopt the Paterson-Stockmeyer algorithm~\cite{Paterson1973on}, which reduces the multiplicative complexity of polynomial evaluation.
While this method approximates the exponential function with high precision over a wide interval, it requires a relatively high multiplicative depth, calculated as $\lceil \log_2(d+1) \rceil + k + 1$, due to the combination of polynomial evaluation and the subsequent $k$ squaring steps required to recover the original scale.

\paragraph{Limit Approximation.}
We implement this as $\mathsf{AExp}(x) = (1 + x/2^k)^{2^k}$.
This approach is computationally efficient, involving only a single scalar multiplication followed by $k$ repeated squarings.
Consequently, it consumes less multiplicative depth, requiring only $k+1$.
However, compared to the Chebyshev method, the Limit approximation is less robust, as the approximation error tends to increase rapidly when the input values fall outside the target domain.

The specific values of the parameter $k$ utilized in our experiments are detailed in Table~\ref{tab:k_values}. In the reported results of Figure~\ref{fig:llama_depth} and Figure~\ref{fig:vit_depth}, we select the approximation method that yields better performance.

\begin{table}[ht]
    \centering
    \caption{Specific values of $k$ employed for the CGF-softmax standard configuration across different models and datasets. All reported values achieve inference accuracy within \SI{1}{\%} of the pre-trained exact softmax baseline.}    \label{tab:k_values}
    \resizebox{0.9\linewidth}{!}{
    \begin{tabular}{lccccccc}
        \toprule
        \textbf{Model} & \multicolumn{3}{c}{\textbf{LLaMA-3.2-1B}} & \textbf{ViT-Base} & \textbf{DeiT-Base} & \textbf{ViT-Tiny} & \textbf{DeiT-Tiny} \\
        \cmidrule(lr){2-4} \cmidrule(lr){5-8}
        \textbf{Dataset} & \small{Clinc150} & \small{Banking77} & \small{SST-2} & \multicolumn{4}{c}{ImageNet-1k} \\
        \midrule
        \textbf{Chebyshev approx.} & 3 & 3 & 4 & 6 & 3 & 1 & 1 \\
        \textbf{Limit approx.} & 6 & 6 & 7 & 8 & 6 & 8 & 7 \\
        \bottomrule
    \end{tabular}
    }
\end{table}

\subsection{Low-degree Variant}
\label{appx:low_degree}
The low-degree variant is designed for scenarios requiring maximum efficiency with minimal multiplicative depth.
In this approach, the approximation is applied \emph{in-training}; the exact exponential function is replaced by a simplified polynomial approximation \emph{before} the training phase begins.
We utilize a Taylor expansion of degree $d$ centered at $x_0$, formulated as:
\begin{equation}
    \mathsf{AExp}(x) = \sum_{i=0}^{d} \frac{e^{x_0}}{i!} (x - x_0)^i
\end{equation}
A key advantage of this method is that it eliminates the need for domain scaling, thereby avoiding the additional overhead associated with squaring operations.
Total multiplicative depth required for this operation is $\lceil \log_2(d+1) \rceil$.

\section{Training Details} \label{appx:search_space}
BPMax employs a knowledge distillation (KD) framework, where the original pre-trained model serves as the \emph{teacher} and the model with the replaced function acts as the \emph{student}.
We adopt the same KD framework to \cgfsoftmax for fair comparison.

\paragraph{Search Space and Configurations.}
To ensure a rigorous evaluation, we conducted an extensive hyperparameter sweep for both LLaMA-3.2-1B and ViT/DeiT.
Table~\ref{tab:configurations} details the specific search grids.
For BPMax, which replace softmax with simplified polynomial alternatives $(x+c)^p$, we significantly expanded the search space beyond the original study~\cite{Park2025Powerformer}--which limited the search to shift variable $c \in \{1,3,5,7\}$ and degrees $p \in \{1,3,5,7\}$--to guarantee optimal baselines.

We set the maximum training budget for BPMax to 20 epochs while applying early stopping to all experiments (including CGF-softmax).
Under these conditions, CGF-softmax demonstrated significantly faster training speeds than the baseline: LLaMA-3.2-1B required only 5 epochs.
For ViT/DeiT models, we employed \emph{Attention-only Tuning}—a strategy where we freeze all model parameters and exclusively update the weights associated with the attention mechanism.
This approach provided the dual benefit of minimizing the training process and facilitating effective accuracy recovery.

\paragraph{Fixed Configurations.}
Regarding input specifications, we set the sequence length to $n=197$ for ViT/DeiT models and the maximum sequence length to 128 for LLaMA-3.2-1B.
For the KD objective, we utilized a combination of Cross-Entropy loss and KL-Divergence loss, setting the temperature to $\tau=2.0$ and the distillation weight to $\alpha=0.5$. 
We fixed the random seed to 42 across all experimental runs.

\begin{table}[ht]
    \centering
    \small
    \caption{Configuration search space and fixed settings for large language models and Vision Transformers. In the Teacher Model row, `Self' refers to the original backbone prior to replacing the softmax, while `ViT-Large (85.84\%)' denotes the pre-trained model sourced from the \texttt{timm} library.}
    \label{tab:configurations}
    \renewcommand{\arraystretch}{1.3} 
    \resizebox{0.9\linewidth}{!}{
    \begin{tabular}{p{4.0cm}cc}
        \toprule
        \textbf{Configuration} & \textbf{LLaMA-3.2-1B} & \textbf{Vision Transformers (ViT/DeiT)} \\
        \midrule
        \multicolumn{3}{l}{\emph{\textbf{CGF-softmax (low degree variant)}}} \\
        \quad Taylor Degree ($d$) & \{$1$\} & \{3, 5\} \\
        \quad Expansion Point ($x_0$) & \{$-3$\} & \{$-10$\} \\
        \midrule
        \multicolumn{3}{l}{\emph{\textbf{BPMax}}} \\
        \quad Degree ($p$) & \{1, 3, 5, 7\} & \{1, 3, 5, 7, 9\} \\
        \quad Scaling ($c$) & \{1, 3, 5, 7, 20\} & \{1, 3, 5, 7, 9, 20, 40, 60\} \\
        \midrule
        \multicolumn{3}{l}{\emph{\textbf{Optimization}}} \\
        \quad Epochs (CGF-softmax) & \{5\} & \{15\} \\
        \quad Epochs (BPMax) & \{20\} & \{20\} \\
        \quad Batch Size & \{32, 64\} & \{64, 128\} \\
        \quad Learning Rate & \{9e-6, 1e-5, 3e-5, 5e-5\} & \{1e-5, 5e-5, 1e-4\} \\
        \midrule
        \multicolumn{3}{l}{\emph{\textbf{Regularization \& Strategy}}} \\
        \quad Drop Path Rate & N/A & \{0.0, 0.1, 0.2\} \\
        \quad Attention-only Tuning & \{Off\} & \{On, Off\} \\
        \midrule
        \multicolumn{3}{l}{\emph{\textbf{Knowledge Distillation}}} \\
        \quad Teacher Model & \{Self\} & \{Self, ViT-Large (85.84\%)\} \\
        \bottomrule
    \end{tabular}
    }
\end{table}

\section{Alternative Training Pipeline: Replace Softmax Before Supervised Fine-Tuning} \label{appx:alt_pipeline}

\begin{table}[b]
    \centering
    \caption{Accuracy and training configuration comparison for LLaMA-3.2-1B downstream adaptation.}
    \label{tab:alt_pipeline}
    \resizebox{\linewidth}{!}{
    \begin{tabular}{llccccccc}
        \toprule
        \multirow{2}{*}{\textbf{Method}} & \multicolumn{4}{c}{\textbf{Training Settings}} & \multicolumn{3}{c}{\textbf{Evaluation Accuracy}} \\
        \cmidrule(lr){2-5} \cmidrule(lr){6-8}
        & \textbf{Base Model} & \textbf{Softmax} & \textbf{Training Method} & \textbf{Epochs} & \textbf{Clinc150} & \textbf{Banking77} & \textbf{SST-2} \\
        \midrule
        Exact softmax baseline & HF Pre-trained & Exact & SFT & 3 & 86.67\% & 92.99\% & 93.11\% \\
        \cgfsoftmax pipeline & Exact softmax baseline & \cgfsoftmax & KD & 5 & 87.02\% & 92.86\% & 92.43\% \\
        BPMax pipeline & Exact softmax baseline & BPMax & KD & 20 & 81.40\% & 91.13\% & 87.15\% \\
        \midrule
        \textbf{Alternative pipeline} & \textbf{HF Pre-trained} & \textbf{\cgfsoftmax} & \textbf{SFT (LoRA)} & \textbf{3} & \textbf{87.09\%} & \textbf{92.04\%} & \textbf{93.00\%} \\
        \bottomrule
    \end{tabular}
    }
\end{table}
When evaluating large language models (LLMs) like LLaMA-3.2-1B on specific downstream tasks, applying Supervised Fine-Tuning (SFT) is standard practice. 
In our main experiments (Figures~\ref{fig:llama_depth} and \ref{fig:vit_depth}), the ``exact softmax'' baselines were established by performing SFT directly on the pre-trained models downloaded from Hugging Face (HF).
Subsequently, to ensure a fair evaluation of the softmax replacement methods, we replaced the softmax function in these already fine-tuned models and conducted an additional Knowledge Distillation (KD) phase (as detailed in Appendix~\ref{appx:search_space}) to create the BPMax and \cgfsoftmax models.
While \cgfsoftmax inherently requires a significantly lower training cost than BPMax during this KD phase, this section introduces an alternative pipeline that removes the \emph{additional} adaptation training cost. 

In this alternative pipeline, we replace the exact softmax with \cgfsoftmax immediately on the base pre-trained model, \emph{prior} to any task-specific tuning. 
We then proceed directly with a standard Low-Rank Adaptation (LoRA) \cite{Hu2022lora} SFT phase. 
As detailed in Table~\ref{tab:alt_pipeline}, this unified approach achieves an accuracy within \SI{1}{\%} of the exact softmax baseline using the exact same number of training epochs.
Because evaluating downstream tasks fundamentally requires an SFT phase, integrating \cgfsoftmax directly into this mandatory process eliminates the need for any subsequent KD steps.
Consequently, adapting the model to \cgfsoftmax incurs virtually zero extra training overhead.


\end{document}